\documentclass[10pt]{NSP1}
\usepackage{url,floatflt}
\usepackage{helvet,times}
\usepackage{psfig,graphics}
\usepackage{mathptmx,amsmath,amssymb,bm}
\usepackage{float}
\usepackage[bf,hypcap]{caption}

\emergencystretch=\maxdimen
\def\no{\noindent}

\def\firstpage{1}
\def\thevol{2}
\def\thenumber{1}
\def\theyear{2015}

\newtheorem{algorithm}{Algorithm}

\begin{document}

\titlefigurecaption{{\large \bf \rm Journal of Statistics Applications \& Probability Letters }\\ {\it\small An International Journal}}

\title{Comparing Seventeen Interval Estimates for a Bivariate Normal Correlation Coefficient}

\author{Mohammad Reza Kazemi\hyperlink{author1}
{$^1$} and Ali Akbar Jafari\hyperlink{author2}
{$^2$}}
\institute{$^1$Department of Statistics, Fasa University, Fasa, Iran\\
$^2$Department of Statistics, Yazd University, Yazd,  Iran}

\titlerunning{Comparing Seventeen Interval Estimates for  Correlation Coefficient}
\authorrunning{M. R. Kazemi, A. A. Jafari}

\mail{kazemi@fasau.ac.ir}

\received{2014}
\revised{2014}
\accepted{...}
\published{...}

\abstracttext{In this paper, we consider the problem of constructing confidence interval for the correlation coefficient in a bivariate normal distribution. For this problem, we found fifteen approaches in literatures. Also, we have proposed a generalized confidence interval and a parametric bootstrap confidence interval. The coverage probabilities and expected lengths of these seventeen approaches are evaluated and compared via simulation study. In addition, robustness of the methods is considered in the comparisons by the non-normal distributions. Two real examples are given to illustrate the approaches.}


\maketitle

\section{Introduction}

A certain departure from stochastic independence between two random variables is assessed by correlation and a well-known measure of linear association between two random variables is the (Pearson product-moment) correlation coefficient. One can investigate the applications of the correlation coefficient in all fields of sciences such as engineering, medicine, psychology, biology and so on.

In a bivariate normal distribution,
\cite{fisher-15}
proposed two expressions for the exact density function of sample correlation coefficient, and
\cite{hotelling-53}
proposed an expression in term of hypergeometric functions.
\cite{ch-li-pa-08}
obtained and tabulated the critical values for the exact test of hypothesis about the correlation coefficient using these expressions.
\cite{shieh-06}
provided an expression for the distribution of sample correlation coefficient using a theorem in linear regression.

The well-known z-transform of
\cite{fisher-21}
is a usual method for inference about the correlation coefficient and constructing confidence interval for this parameter. After that, many authors investigated other approximations and confidence intervals:
\cite{hotelling-53}
offered four modifications of Fisher's z-transformation.
\cite{ruben-66}
obtained a simple approximate normalization for the correlation coefficient in normal samples. For testing,
\cite{samiuddin-70}
proposed a test statistic, and
\cite{muddapur-88}
developed this test statistic.
\cite{muddapur-88}
also obtained a test statistic based on the F distribution and used it to construct a confidence interval.
\cite{jeyaratnam-92}
gave another form of this confidence interval.
\cite{su-wo-07}
used method of signed log-likelihood ratio statistic and derived two confidence intervals. Using the concept of generalized pivotal variable,
\cite{kr-xi-07}
developed a generalized confidence interval. Using Cornish-Fisher expansions,
\cite{wi-na-12}
gave two approximate confidence intervals.
\cite{ha-pr-11}
derived an alternative estimator of Pearson's correlation coefficient in terms of the ranges. Then they found an approximate confidence interval for this parameter and provided a new approximation for the density function of sample correlation coefficient.

Here, we consider inference about the correlation coefficient parameter and propose a generalized pivotal variable and a parametric bootstrap (PB) approach for constructing confidence interval for this parameter.
This paper is organized as follows: In Section \ref{sec.CI}, we first review the existing approaches to construct confidence interval for the correlation coefficient parameter. Then, we derive a generalized pivotal variable and a PB approach for this problem. Simulation studies are performed in Section \ref{sec.sim} to evaluate and compare the coverage probability and expected length of these approaches. Two real examples are given in Section \ref{sec.ex}.
The paper is concluded in Section \ref{sec.con}.

\section{Confidence intervals for $\rho$}
\label{sec.CI}

Let $(X_{11},X_{21})',\dots ,(X_{1n},X_{2n})'$ be a random sample from a bivariate normal distribution with mean vector
${\boldsymbol \mu }=({\mu }_1,{\mu }_2 )'$ and variance-covariance matrix
\[\Sigma =\left[ \begin{array}{cc}
{\sigma }^2_1 & \rho {\sigma }_1{\sigma }_2 \\
\rho {\sigma }_1{\sigma }_2 & {\sigma }^2_2 \end{array}
\right],\]
where $\rho $ is the correlation coefficient between first and second components. The maximum likelihood estimation for $\Sigma $ is $S=(S_{ij})=\frac{1}{n}\sum^n_{i=1}{({{\boldsymbol X}}_i-\bar{{\boldsymbol X}})({{\boldsymbol X}}_i-\bar{{\boldsymbol X}})'}$, where
\[S_{ii}=S^2_i=\frac{1}{n}\sum^n_{j=1}{{(X_{ij}-{\bar{X}}_i)}^2},\ \ \ \ i=1,2,\ \ \ \ \ \ \ \ \ \ S_{12}=\frac{1}{n}\sum^n_{j=1}{(X_{1j}-{\bar{X}}_1)(X_{2j}-{\bar{X}}_2)},\]
and ${\bar{X}}_i=\frac{1}{n}\sum^n_{j=1}{X_{ij}}$, $i=1,2$. Therefore, the maximum likelihood estimation of $\rho $ is $R=\frac{S_{12}}{S_1S_2}$.

Using the expressions for density function of $R$ given by
\cite{fisher-15}
or by
\cite{hotelling-53},
one can construct an exact $100(1-\alpha)\%$ confidence interval for $\rho$ by numerical solving the following equations
\[\int^L_{-1}{f(r;\rho)dr=\frac{\alpha }{2}},\ \ \ \ \ \ \ \ \ \ \int^1_U{f(r;\rho)dr=\frac{\alpha }{2}},\]
where $f(r;\rho)$ is the density function of sample correlation coefficient. For more details on the form of the $f(r;\rho)$, one can refer to
\cite{hotelling-53}.
However, this is a very difficult method and needs solving a complex integral. Therefore, some approximations methods are proposed to construct confidence interval for $\rho $.

In this Section, we first review the existing methods to construct confidence interval for $\rho $. Then we will propose a generalized pivotal variable and a PB approach.

\subsection{Fisher's z-transformation}
The Fisher's z-transformation is the most well-known and popular
approximation for the sample correlation coefficient $R$. It is also known as the variance stabilizing transformation.
\cite{fisher-21}
showed that
\begin{equation}\label{eq.zt}
Z=\frac{1}{2}{\log (\frac{1+R}{1-R})}={\tanh }^{-1} (R),
\end{equation}
has an asymptotic normal distribution with mean
$\zeta =\frac{1}{2}\log  (\frac{1+\rho }{1-\rho })={\tanh^{{-1}}(\rho ) }$ and variance $1/(n-3)$. Therefore, an approximate $100(1-\alpha)\%$ confidence interval for $\rho $ is given by
\[\left(\tanh(Z-\frac{Z_{\alpha/2}}{\sqrt{n-3}} ), \tanh  (Z+\frac{Z_{\alpha/2}}{\sqrt{n-3}})\right),\]
where $Z_\gamma$ is the  $\gamma $th upper quantile of the standard normal distribution.

\subsection{Hotelling's approximations}

\cite{hotelling-53}
gave four modifications of Fisher's z-transformation as
\begin{eqnarray*}
&&Z_1=Z-\frac{7Z+R}{8(n-1)},\ \ {\zeta }_1=\zeta -\frac{7\zeta +\rho }{8(n-1)},\\
&&Z_2=Z-\frac{7Z+R}{8(n-1)}-\frac{119Z+57R+3R^2}{384{(n-1)}^2},\ \
{\zeta }_2=\zeta -\frac{7\zeta +\rho }{8(n-1)}-\frac{119\zeta +57\rho +3{\rho }^2}{384(n-1)^2},\\
&&Z_3=Z-\frac{3Z+R}{4(n-1)},\ \ {\zeta}_3=\zeta -\frac{3\zeta +\rho }{4(n-1)},\\
&&
Z_4=Z-\frac{3Z+R}{4(n-1)}-\frac{23Z+33R-5R^2}{96{(n-1)}^2},\ \
{\zeta }_4=\zeta -\frac{3\zeta +\rho }{4(n-1)}-\frac{23\zeta +33\rho -5{\rho }^2}{96(n-1)^2},
\end{eqnarray*}
and showed that $Z_i$, $i=1,\dots ,4$, are distributed as normal distribution with means ${\zeta }_i$, $i=1,\dots ,4$, and variances $1/(n-1).$ Therefore, we can construct four confidence intervals for $\rho$ based on these approximations. However, there is no closed form for each of these confidence intervals, and they are obtained numerically.

\subsection{Ruben's approximation}
\cite{ruben-66}
showed that
\[Z_{hr}=\frac{{\left(\frac{2n-5}{2}\right)}^{\frac{1}{2}}\tilde{R}-{\left(\frac{2n-3}{2}\right)}^{\frac{1}{2}}\tilde{\rho }}{{\left(1+\frac{1}{2}({\tilde{R}}^2+{\tilde{\rho }}^2)\right)}^{\frac{1}{2}}},\]
is asymptotically distributed as a standard normal distribution, where
$\tilde{R}=\frac{R}{\sqrt{1-R^2}}$ and
$\tilde{\rho }=\frac{\rho }{\sqrt{1-{\rho }^2}}$.
Therefore, we can construct a confidence interval for $\rho $ numerically based on this approximation.

\subsection{Muddapur's methods}
\cite{samiuddin-70}
proposed a test statistic for testing $\rho =0$ as
\[t=\frac{\sqrt{n-2}(R-\rho)}{\sqrt{\left(1-\rho^2\right)\left(1-R^2\right)}},\]
and showed that this statistic has an approximate $t$ distribution with $n-2$ degrees of freedom for moderately large $n$.
\cite{muddapur-88}
developed this test statistic as
\[t=\frac{\sqrt{n-2}(R-\rho b)}{\sqrt{(1-{\rho }^2)(1-R^2)}},\]
where $b=(S^2_1+S^2_2)/\sqrt{4S^2_1S^2_2}$. He showed that this test statistic has a $t$ distribution with $n-2$ degrees of freedom. This test statistic is a likelihood ratio test obtained by
\cite{anderson-58}
for testing the hypothesis $\rho =\rho_0$. Therefore, one can construct a confidence interval for the parameter $\rho $.

\cite{muddapur-88}
also considered the test statistic
\[f=\frac{(1+R)(1-\rho)}{(1-R)(1+\rho)},\]
and showed that it has an approximate F distribution with
$ n-2$ and $ n-2$ degrees of freedom. Note that $f$ is related to Fisher's z-transform through the one to one relationship as
\[\log(f)=2(Z-\zeta ).\]
Therefore, a $100(1-\alpha)\%$ confidence interval for $\rho $ is
\[\left(\frac{(1+F_{\alpha /2})R+(1-F_{\alpha /2})}{(1+F_{\alpha /2})+(1-F_{\alpha /2})R},\ \ \frac{(1+F_{\alpha /2})R-(1-F_{\alpha /2})}{(1+F_{\alpha /2})-(1-F_{\alpha /2})R}\right),\]
where $F_\gamma$ is the  $\gamma $th upper quantile of the F distribution with $n-2$ and $n-2$ degrees of freedom.
\cite{jeyaratnam-92}
gave another form of this confidence interval as
\[\left(\frac{R-w}{1-Rw},\frac{R+w}{1+Rw}\right),\]
where
\[w=\frac{t_{(n-2,\alpha /2)}/\sqrt{n-2}}{{\left(1+{(t_{(n-2,\alpha /2)})}^2/(n-2)\right)}^{1/2}}.\]

\subsection{Signed log likelihood method}

\cite{su-wo-07}
used the method of signed log-likelihood ratio statistic and its modification to construct two confidence intervals for $\rho$. Let
$\ell ({\boldsymbol \theta })$, where $\boldsymbol \theta =({\mu }_1,{\mu }_2,{\sigma }_1,{\sigma }_2,\rho )'$,  be the log-likelihood function of a sample from bivariate normal distribution, and suppose that $\hat{{\boldsymbol \theta }}$ be the maximum likelihood estimator of the parameter
${\boldsymbol \theta }$. Let
$\psi ({\boldsymbol \theta })$ (here $\rho $) be a scalar parameter of interest and
${\hat{{\boldsymbol \theta }}}_{\psi }$ be the constrained maximum likelihood estimator which can be obtained by maximizing
$\ell ({\boldsymbol \theta })$ subject to the constrain
$\psi ({\boldsymbol \theta })=\psi $. Thus
\[D(\psi)={\rm sign}(\hat{\psi }-\psi )\left(2\ell (\hat{{\boldsymbol \theta }})-2\ell
({\hat{{\boldsymbol \theta }}}_\psi )\right)^{1/2},\]
is the signed log-likelihood ratio statistic and asymptotically is distributed as
$N(0,1)$. Therefore, a $100(1-\alpha)\%$ confidence interval for $\psi $ is $\{\psi :|D(\psi)|<Z_{\alpha /2}\}$.

The accuracy of $D(\psi )$ is $O(n^{-\frac{1}{2}})$, and to improve the accuracy of $D(\psi)$, the modified signed log-likelihood ratio statistic was proposed by
\cite{barndorff-86,barndorff-91}.
Consider
\[D^*(\psi )=D(\psi )-\frac{1}{D(\psi )}{\log  \frac{D(\psi )}{Q(\psi )} },\]
where
\[Q(\psi )=(\hat{\psi }-\psi){\left\{\frac{\left|j_{{\boldsymbol \theta }{{\boldsymbol \theta }}'}(\hat{{\boldsymbol \theta }})\right|}{\left|j_{{\boldsymbol \lambda }{{\boldsymbol \lambda }}'}({\hat{{\boldsymbol \theta }}}_{\psi })\right|}\right\}}^{1/2},\]
and $j_{{\boldsymbol \theta }{{\boldsymbol \theta }}'}(\hat{{\boldsymbol \theta }})$ and
 $j_{{\boldsymbol \lambda }{{\boldsymbol \lambda }}'}({\hat{{\boldsymbol \theta }}}_{\psi })$ is the observed information matrix evaluated at $\hat{{\boldsymbol \theta }}$ and observed nuisance information matrix evaluated at ${\hat{{\boldsymbol \theta }}}_{\psi }$, respectively. For more details on the modified signed log-likelihood ratio statistic $D^*(\psi)$ for inference about correlation coefficient, $\rho $, reader can refer to
\cite{su-wo-07}.
Thus, a $100(1-\alpha)\%$ confidence interval for $\psi $ is $\{\psi :|D^*(\psi )|<Z_{\alpha /2}\}$. Note that the calculations  are not simple in this method.

\subsection{Krishnamoorthy and Xia's Method}

The concepts of generalized pivotal variable and generalized confidence interval are defined by
\cite{weerahandi-93}
and are used by some authors in many statistical problems. For more information about this concepts, see the book by
\cite{weerahandi-95}.

\cite{kr-xi-07}
proposed a generalized confidence interval for $\rho $ by developing a generalized pivotal variable as
\[G_{\rho 1}=\frac{\tilde{r}V_{22}-V_{21}}{\sqrt{{(\tilde{r}V_{22}-V_{21})}^2+V^2_{11}}},\]
where $\tilde{r}=\frac{r}{\sqrt{1-r^2}}$, and $V^2_{11}$, $V^2_{22}$, and $V_{21}$ are independent random variables with
${\chi }^2_{(n-1)}$, ${\chi }^2_{(n-2)}$, and $N(0,1)$, respectively. The confidence interval for $\rho $ can be constructed by using a Monte Carlo simulation. 

\subsection{Withers and Nadarajah's methods}

\cite{wi-na-12}
considered Cornish-Fisher expansions for the distribution of $Y_m=m^{\frac{1}{2}}(\hat{\theta}-\theta)a_{12}$ as
\[P^{-1}_m(t)={\Phi }^{-1}(t)+\sum^{\infty }_{j=1}{m^{-j/2}{{\rm g}}_j({\Phi }^{-1}(t))},\]
where $P_m(t)$ and $\Phi (t)$ are the distribution function of $Y_m$ and a standard normal distribution, respectively,
and ${{\rm g}}_j$ is certain polynomials in $x$. They considered
${\boldsymbol \mu } ={\boldsymbol 0}$, and therefore, gave the estimation of $\rho $ as
\[\hat{\rho }=\frac{\sum^n_{j=1}{X_{i1}X_{i2}}}{\sqrt{(\sum^n_{j=1}{X^2_{i1}})(\sum^n_{j=1}{X^2_{i2}})}}.\]
Then, using $Y_m$ with the following values
\begin{eqnarray*}
&&m=n,\ \ \ \theta =\tanh^{- 1} (\rho ),\ \ \ \hat{\theta}={\tanh }^{-1} (\hat{\rho}),\ \ \ \ \ a_{21}=1,\\
&&{{\rm g}}_1(x)=\frac{\rho }{2}+{\rho }^3\frac{(x^2-1)}{6},\ \ \ \ \
{{\rm g}}_2(x)=\frac{x^3}{12}+\frac{x}{4}-{\rho }^2\frac{x}{4}-\rho ^6\frac{(2x^3-5x)}{36},
\end{eqnarray*}
they proposed two confidence intervals for $\rho $. These confidence intervals are obtained numerically.

\subsection{Haddad and Provost's method}

Let $D^+=\sum^n_{i=1}{{(X^*_{i1}+X^*_{i2})}^2}$ and
$D^-=\sum^n_{i=1}{{(X^*_{i1}-X^*_{i2})}^2}$, where $X^*_{ij}=\frac{X_{ij}-{\bar{X}}_i}{S_i}$, $i=1,2$, $j=1,\dots ,n$, are the standard values.
\cite{ha-pr-11}
proposed an approximately $100(1-\alpha)\%$ confidence interval for $\rho$ as
\[\left(\frac{D^+-D^-F^*_{\alpha /2}}{D^++D^-F^*_{\alpha /2}},\frac{D^+-D^-F^*_{1-\alpha /2}}{D^++D^-F^*_{1-\alpha /2}}\right),\]
where $F^*_{\gamma }$ is the  $\gamma $th upper quantile of the F distribution with $n-1$ and $n-1$ degrees of freedom.

\subsection{A new generalized confidence interval}

Using the concept of generalized confidence interval, we construct a new confidence interval for the correlation coefficient parameter.
Let $A=nS$, and $a=(a_{ij})$ be an observed value of matrix $A$. Therefore $A\sim W(n-1,\Sigma )$ and also is distributed as $\sum^{n-1}_{i=1}{{{\boldsymbol Z}}^*_i{{\boldsymbol Z}}^{*'}_i}$, where ${{\boldsymbol Z}}^*_i$ has multivariate normal distribution with mean vector
${\boldsymbol \mu } ={\boldsymbol 0}$ and variance-covariance matrix $\Sigma $. Then, $a^{-\frac{1}{2}}Aa^{-\frac{1}{2}}$ is distributed as $\sum^{n-1}_{i=1}{{{\boldsymbol Z}}^{**}_i{{\boldsymbol Z}}^{**'}_i}$, where ${{\boldsymbol Z}}^{**}_i$ has multivariate normal distribution with mean vector ${\boldsymbol \mu } ={\boldsymbol 0}$\textbf{ }and variance-covariance matrix $a^{-\frac{1}{2}}\Sigma a^{-\frac{1}{2}}$. So, $a^{-\frac{1}{2}}Aa^{-\frac{1}{2}}$ follows a Wishart distribution with $n-1$ degrees of freedom and scale parameter matrix $a^{-\frac{1}{2}}\Sigma a^{-\frac{1}{2}}$, i.e. $a^{-\frac{1}{2}}Aa^{-\frac{1}{2}}\sim W(n-1,a^{-\frac{1}{2}}\Sigma a^{-\frac{1}{2}})$.
Consequently
\[V^*=(V^*_{ij})=a^{-\frac{1}{2}}{(a^{-\frac{1}{2}}\Sigma a^{-\frac{1}{2}})}^{-\frac{1}{2}}(a^{-\frac{1}{2}}Aa^{-\frac{1}{2}}){(a^{-\frac{1}{2}}\Sigma a^{-\frac{1}{2}})}^{-\frac{1}{2}}a^{-\frac{1}{2}}\sim W(n-1,a^{-1}).\]

The value of $V^*$ at $A=a$ is ${\Sigma }^{-1}$, and for a given $a$, the distribution of $V^*$ does not depend on any unknown parameters. Therefore, $V^*$ is a generalized pivotal variable for ${\Sigma }^{-1}$, and
$$V^{*-1}=\frac{1}{V^*_{11}V^*_{22}-{(V^*_{12})}^2}
\left[ \begin{array}{cc}
V^*_{22} & -V^*_{12} \\
-V^*_{12} & V^*_{11} \end{array}
\right],$$
is a generalized pivotal variable for $\Sigma $. So,
\begin{equation}\label{eq.G2}
G_{\rho 2}=\frac{-V^*_{12}}{\sqrt{V^*_{11}V^*_{22}}},
\end{equation}
is a generalized pivotal variable for the parameter $\rho $. By using the Monte Carlo simulation given in the following algorithm, we can find a confidence interval for $\rho $.

\begin{algorithm}

For given sample covariance matrix, $s$, \\
\textbf{Step 1}. Compute $a$ and $a^{-1}$.\\
\textbf{Step 2}. Generate $V^*\sim W(n-1,a^{-1})$.\\
\textbf{Step 3.} Compute $G_{\rho 2}$ in \eqref{eq.G2}.\\
\textbf{Step 4}. Repeat Step 2 and 3 for a large number of times (say $M=10,000$)\\
Then from these $M$ values, the $100(\alpha /2)$th and $100(1-\alpha /2)$th percentile of $G_{\rho 2}$ is a $100(1-\alpha)\%$ confidence interval for $\rho $.
\end{algorithm}

\subsection{A parametric bootstrap method}

The bootstrap approach is a computer-based method that is applied on the observed data by Monte Carlo simulation
\cite{ef-ti-94}.
Using the PB approach, one can approximate the null distribution of some statistical tests.
This approach was used by some authors in well-known problems like the Behrens-Fisher problem
 \cite{ch-pa-08-behrens},
comparing several normal means
\cite{ch-pa-li-li-10},
ANCOVA with unequal variances
\cite{sa-ja-13},
the equality of coefficients of variation
\cite{ja-ka-13},
and the equality of two log-normal means
\cite{ja-ab-14}.
 Here, we propose a PB confidence interval for $\rho$ using the z-transformation in \eqref{eq.zt}.

\begin{lemma} Let $R$ be the sample correlation coefficient for a bivariate normal distribution with mean vector
${\boldsymbol \mu }$ and variance-covariance matrix $\Sigma $. Then
\[R\sim \frac{\tilde{\rho} V+N}{\sqrt{(\tilde{\rho} V+N)^2+W^2}},\]
where
$\tilde{\rho}=\frac{\rho}{\sqrt{1-{\rho }^2}}$, and
 $V^2$, $W^2$, and $N$ are independent random variables with ${\chi }^2_{(n-1)}$, ${\chi }^2_{(n-2)}$, and $N\left(0,1\right)$.
\end{lemma}

\begin{proof} Since $\Sigma $  is a positive definite matrix, there is a unique lower triangular matrix, $L$, such that $LL'=\Sigma $ (Cholesky decomposition). From
\cite{muirhead-82},
page 99,
\[LCC'L'=A\sim W\left(n-1,\Sigma \right),\]
where
$A=(A_{ij})=nS$, $C=\left[ \begin{array}{cc}
V & 0 \\
N & W \end{array}
\right]$
 and $V^2$, $W^2$, and $N$ are independent random variables with ${\chi }^2_{(n-1)}$, ${\chi }^2_{(n-2)}$, and $N\left(0,1\right)$, respectively.
It can be shown that matrix $L$ has the following form:
\[L=\left[ \begin{array}{cc}
{\sigma }_1 & 0 \\
\rho {\sigma }_2 & {\sigma }_2\sqrt{1-{\rho }^2} \end{array}
\right].\]
Therefore, we have
\[A=\left[
\begin{array}{cc}
{{\sigma }^2_1V}^2 & {{\sigma }_1\sigma_2}{\sqrt{1-\rho ^2}}(\tilde{\rho} V^2+NV) \\
{{\sigma }_1\sigma_2}{\sqrt{1-\rho ^2}}(\tilde{\rho} V^2+NV) &
{{\sigma }^2_2}{(1-{\rho }^2)}[(\tilde{\rho} V+N)^2+W^2]
\end{array}
\right].\]
So, the sample correlation coefficient is distributed as $ \frac{A_{12}}{\sqrt{A_{11}A_{22}}}$, and the proof is completed.
\end{proof}

Consider
$Q={(Z-{{\tanh }^{ - 1} (\rho )})}^2$.
We can approximate the distribution of $Q$ 
using a PB approach as
\begin{equation}\label{eq.QB}
Q^B={(Z^B-\tanh ^{-1} (R))}^2,
\end{equation}
where $Z^B=\frac{1}{2}{\log (\frac{1+R^B}{1-R^B}) }$ with
\begin{equation}\label{eq.RB}
R^B\sim \frac{\tilde{R} V+N}{\sqrt{(\tilde{R} V+N)^2+W^2}},
\end{equation}
where
$\tilde{R}=\frac{R}{\sqrt{1-{R }^2}}$.
Then, the distribution of $Q^B$ provides the PB approximation for the distribution of $Q$,  and a $100(1-\alpha )\%$ PB confidence interval for $\rho $ is
\[\left(\tanh  (Z-\sqrt{q^B_{\alpha }}) , \tanh  (Z+\sqrt{q^B_{\alpha }})\right)\]
where $q^{B}_{\alpha }$ denotes the $(1-\alpha )$th quantile of the distribution of $Q^B$. The value of $q^B_{\alpha }$ can be estimated using a Monte Carlo simulation as follows:

\begin{algorithm} For given sample correlation coefficient, $R=r$,\\
\textbf{Step 1}. Generate artificial sample correlation coefficients $R^B$ in \eqref{eq.RB} and compute $Q^B$ in \eqref{eq.QB}.\\
\textbf{Step 2}. Repeat Step 2 for a large number of times (say $M=10,000$), and from these $M$ values, obtain the empirical distribution of $Q^B$ and its $(1-\alpha )$th quantile as an estimate of $q^B_{\alpha }$.
\end{algorithm}

\begin{remark}
We used the Fisher z-transformation quantity to construct the PB confidence interval for $\rho $. In addition, one can use quantities in Hotelling, Ruben, Muddapur methods, Jeyaratnam, Sign log likelihood, and Withers and Nadarajah methods, and propose other PB confidence intervals.
\end{remark}

\section{Simulation Studies}
\label{sec.sim}

For evaluating the performance of the methods for constructing confidence
intervals for $\rho $, we compare the coverage probabilities and expected lengths using simulation studies. In all cases, we consider 95\% confidence coefficient, and generate 10,000 random samples with size $n$ from a bivariate normal distribution with
 $(\mu_1,\mu_2)=(0,0)$, $(\sigma_1,\sigma_2)=(1,1)$, and $\rho =0.0, 0.1, \dots , 0.9$.

  With respect to different values of $\rho$, the coverage probabilities and expected lengths  are plotted   in Figures \ref{fig.cov}  and \ref{fig.len}, respectively, and we can conclude that:\\
i) For all $n$ and $\rho $, the coverage probabilities of the exact method, Fisher's z-transformation method, Krishnamoorthy's generalized confidence interval, Withers and Nadarajah' methods, and PB confidence interval are close to the confidence coefficient.\\
ii) The coverage probability and expected length of the PB method are close to the coverage probability and expected length of Fisher's z-transformation method. However, PB method is applicable for $n=3$ but Fisher's z-transformation method is not.\\
iii) When the sample size is small the coverage probabilities of Ruben, and Haddad and Provost' methods are larger than the confidence coefficient, and are satisfactory for other sample sizes.\\
iv) The coverage probabilities of confidence intervals based on modified signed log-likelihood ratio and Hotellings' $Z_i$ are satisfactory when $n\ge 10$.\\
v) There are a few situations that, the coverage probabilities of our generalized pivotal approach is less than the confidence coefficient.\\
vi) The expected lengths of all confidence intervals become small when the sample size $n$ or the parameter $\rho $ becomes large.

Also, we performed a simulation study of robustness of the confidence intervals for $\rho $. We consider two cases: First, 10,000 random samples with size $n$ are generated from a bivariate t distribution with 5 degrees of freedom and parameters
$(\mu_1,\mu_2)=(1,2)$, $(\sigma_1,\sigma_2)=(1,3)$,
and $\rho =0,0.6$. The results are given in Table \ref{tab.simt} and we can conclude that the coverage probabilities of all confidence intervals are not satisfactory, and are smaller than the confidence coefficient.

In second robustness study, 10,000 random samples with size $n$ are generated from a bivariate log-normal distribution. A random variable
${\boldsymbol X} =(X_1,X_2)'$
has a bivariate log-normal distribution with parameters
${\boldsymbol \mu }$ and $\Sigma $ if
${\boldsymbol Y} =(\log  (X_1) ,\log  (X_2))'$
has the bivariate normal distribution with mean ${\boldsymbol \mu }$ and covariance matrix $\Sigma $. Note that the correlation coefficient for a bivariate log-normal distribution equals to
\[{\rho }^*=\frac{{\exp  (\rho {\sigma }_1{\sigma }_2)}-1}{\sqrt{(\exp  ({\sigma }^2_1)-1)(\exp  ({\sigma }^2_2) -1)}},\]
and the coverage probability of a confidence interval is the number the cases that the parameter  ${\rho }^*$ lies within the confidence interval. Note that $\rho^*$ and  $\rho $ are close when ${\sigma }_1$ and ${\sigma }_2$ both are small. Here, we consider
$({\mu }_1,{\mu }_2)=(1,2)$, $({\sigma }_1,{\sigma }_2)=(0.1,0.1)$, $\rho =0,0.6$. Therefore, $\rho^*=0, 0.5988$.  The results are given in Table \ref{tab.simlog}, and we can conclude the similar results to normal case. Note that when ${\sigma }_i$'s are not equal or when ${\sigma }_i$'s are large, all confidence intervals are unsatisfactory, and their coverage probabilities are smaller than the confidence coefficient.

\section{Two numerical examples}
\label{sec.ex}

In this Section, we illustrate the seventeen confidence intervals for the correlation coefficient $\rho $ using two real examples.

\begin{example}
\cite{levine-99}
 studied the role of nonexercise activity thermogenesis in resistance to fat gain in humans and obtained the following data set. They considered two variables. $X_1$: the increase in energy use (in cal) from activity other than deliberate exercise, and $X_2$: the fat gain (in kg). This data set is also studied by
 \cite{su-wo-07}.
 The sample correlation coefficient equals to -0.7786. The 95\% confidence intervals are given in Table \ref{tab.ex}. We can find that the confidence intervals based on exact method, the confidence interval based on Modified signed log-likelihood ratio
  (proposed by \cite{su-wo-07}),
  and generalized confidence interval
  (proposed by \cite{kr-xi-07})
  are the same.
\end{example}

\begin{example}

The source of the data set is National Center for Education Statistics (http: //nces.ed.gov).
\cite{ch-li-pa-08}
also studied this. The Trial Urban District Assessment (TUDA) is a special project under the US National Assessment of Educational Progress (NAEP). It began assessing performance in selected large urban districts in 2002 and eleven urban school districts participated in the TUDA at grades 4 and 8 in 2005. The sample correlation coefficients between score in Mathematics and Reading for grades 4 and 8 equal to 0.9755 and 0.9738, respectively. The 95\% confidence intervals using different methods for two grades are given in Table \ref{tab.ex}. We can find that the confidence interval based on exact method, and the generalized confidence interval
(proposed by \cite{kr-xi-07})
are same in each grade score.

\end{example}

\begin{figure}
\centering
\includegraphics[scale=0.56]{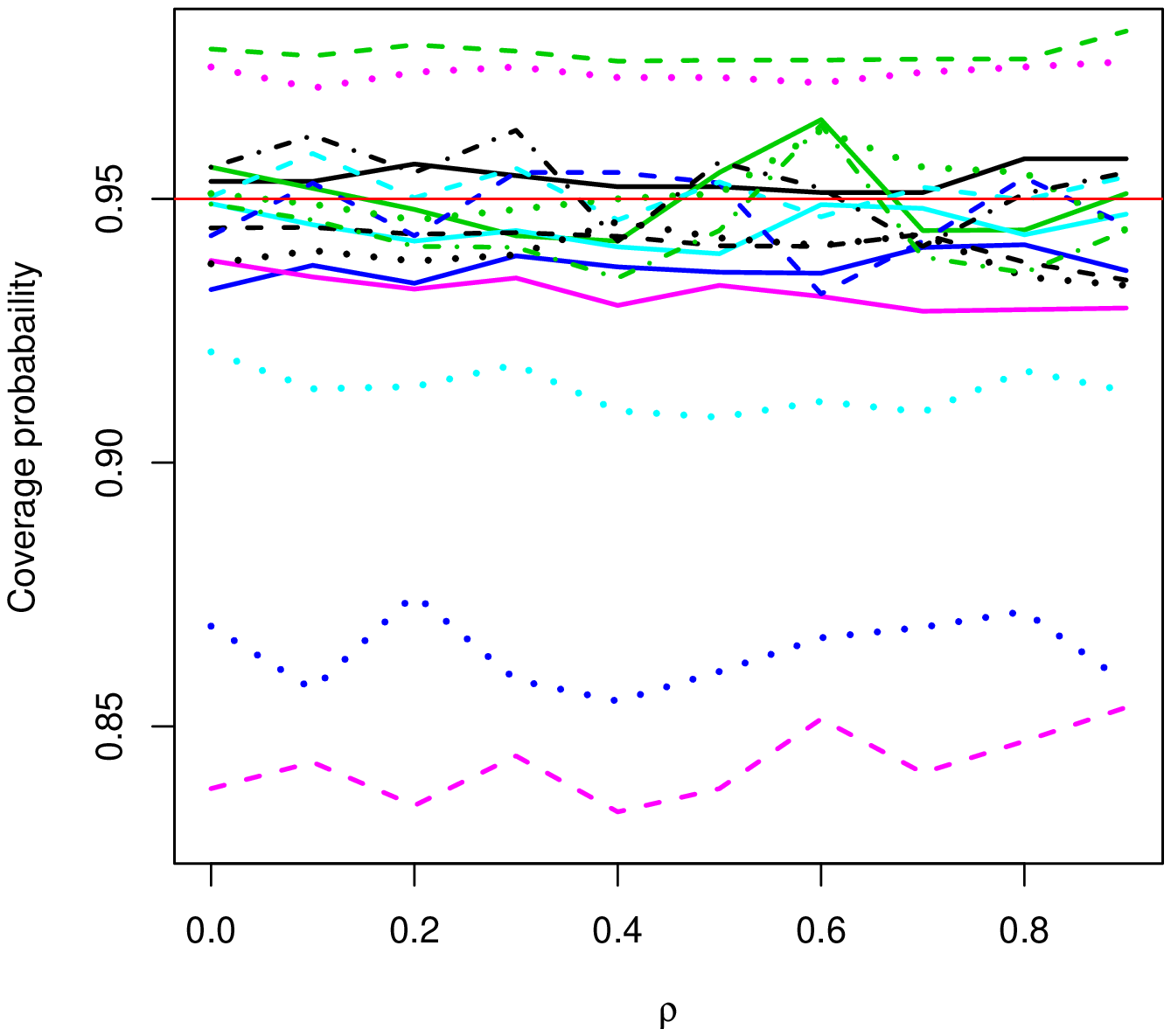}
\includegraphics[scale=0.56]{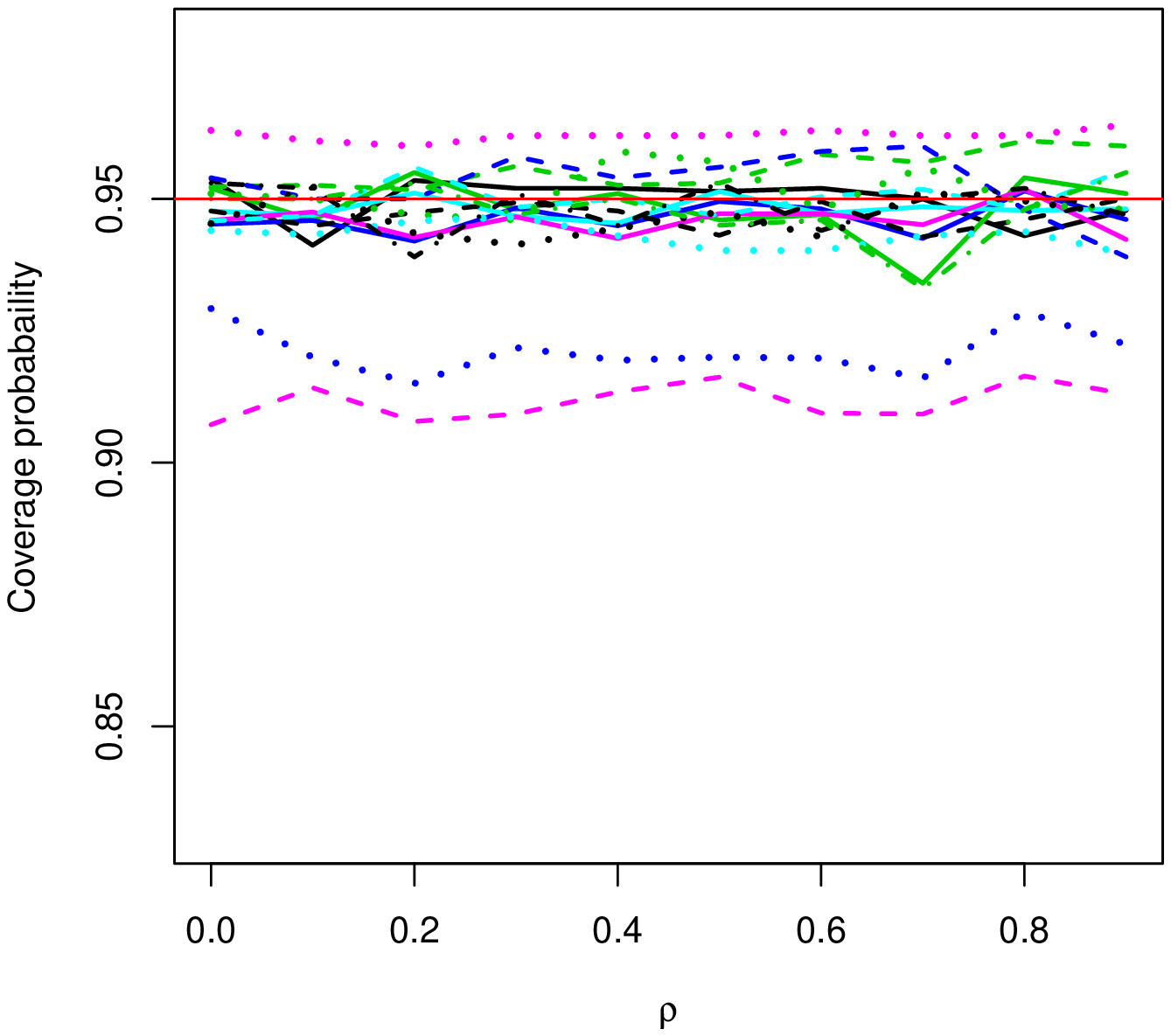}
\includegraphics[scale=0.56]{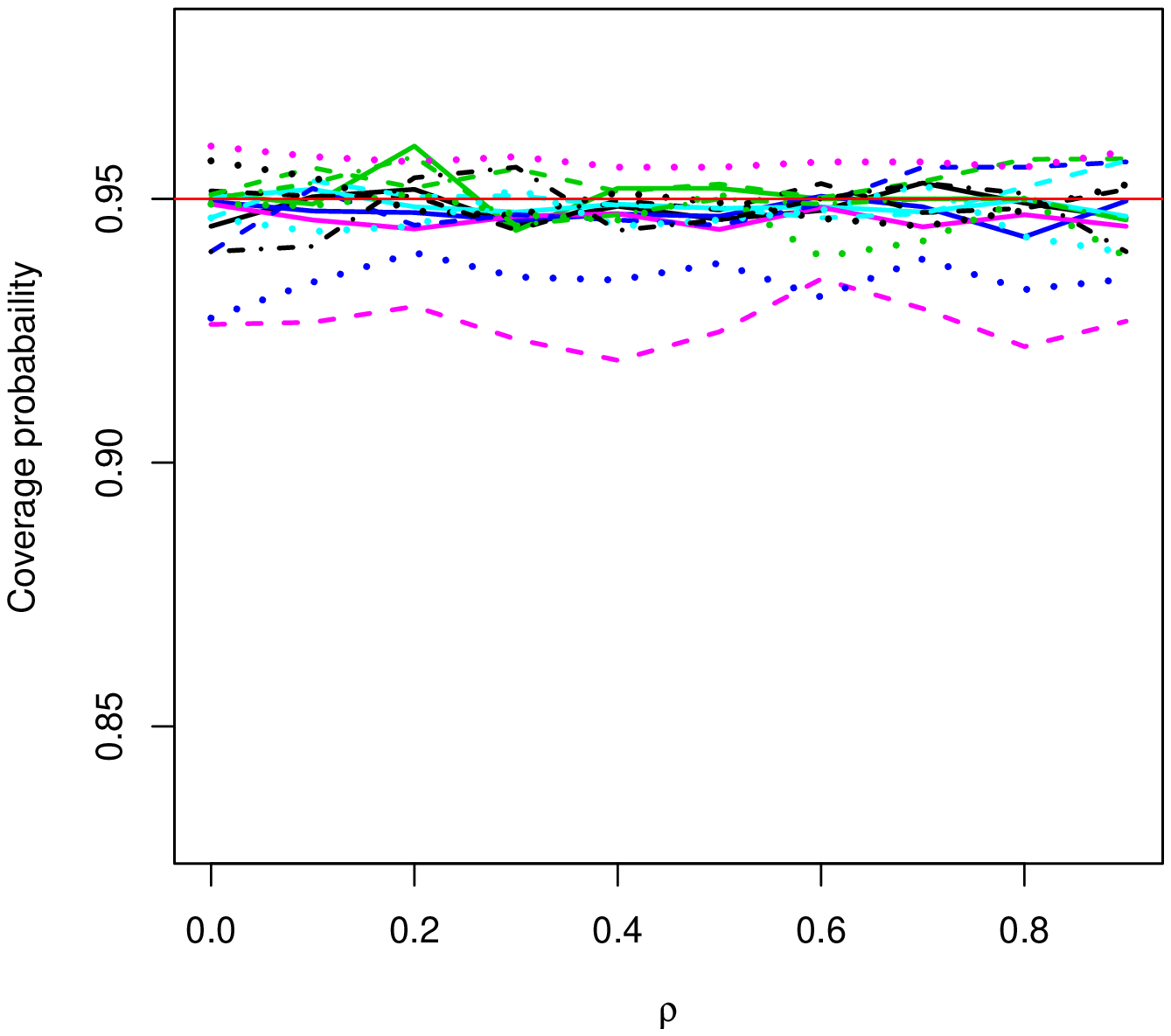}
\includegraphics[scale=0.56]{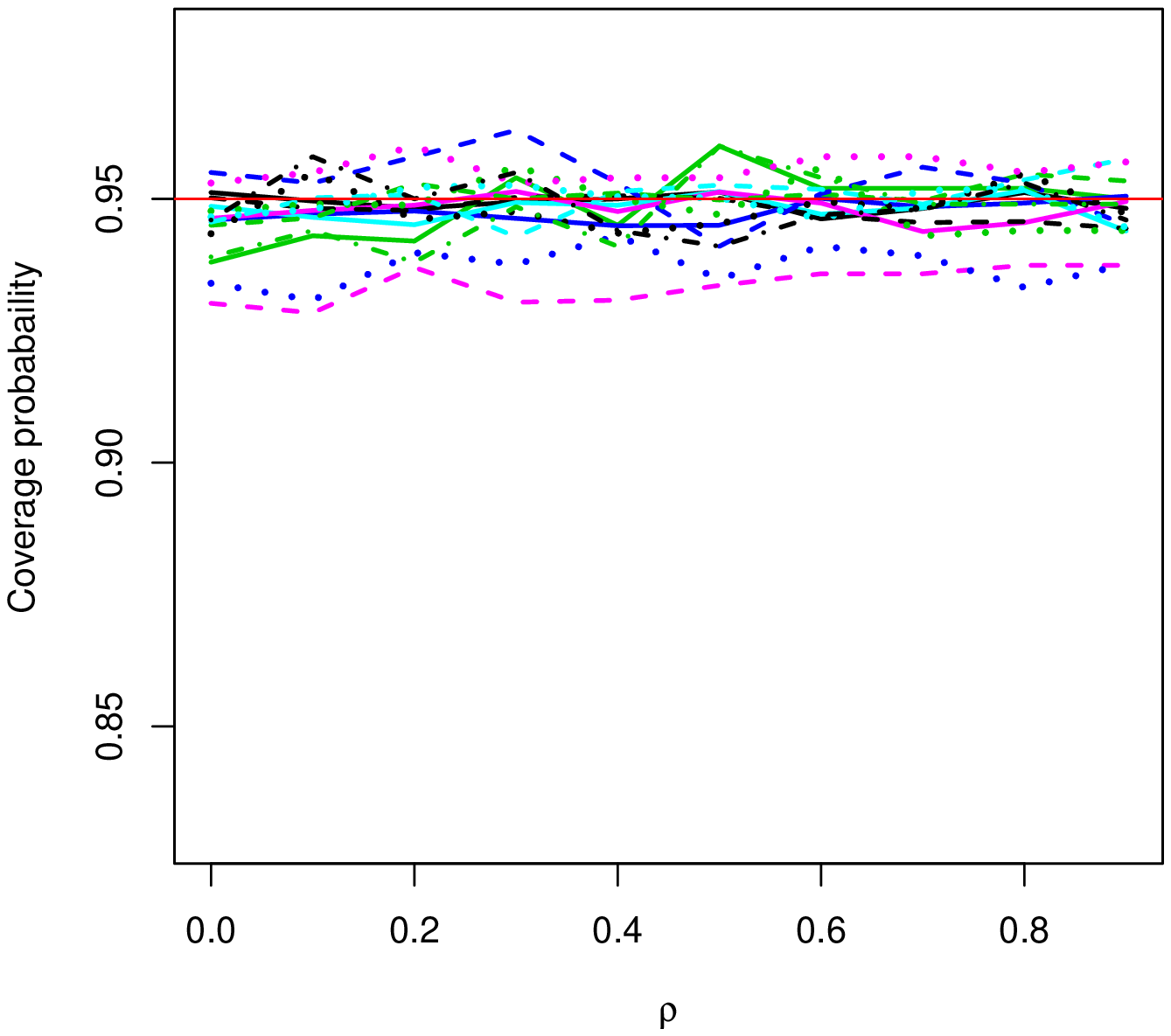}
\includegraphics[scale=0.7]{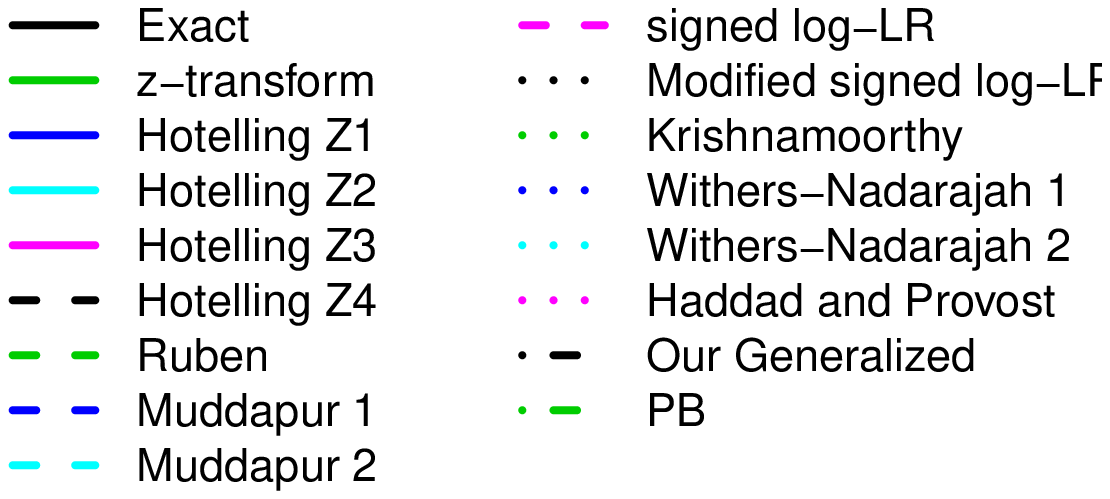}
\vspace{-5cm}
\caption[]{The coverage probabilities of confidence intervals for $\rho$ with $n=5$ (top left), $n=10$ (top right), $n=15$ (bottom left), and $n=20$ (bottom right).}\label{fig.cov}
\end{figure}

\begin{figure}
\centering
\includegraphics[scale=0.56]{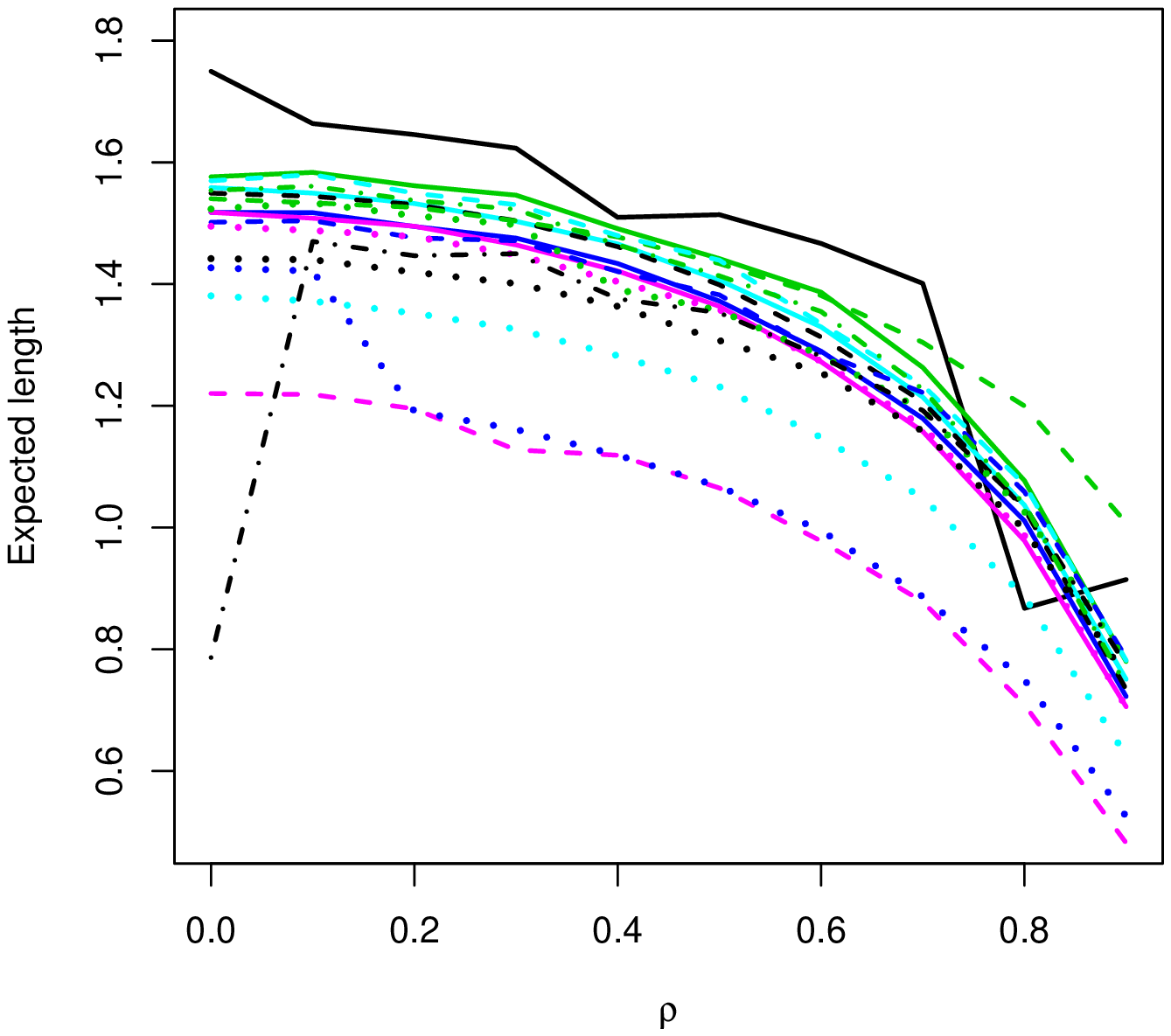}
\includegraphics[scale=0.56]{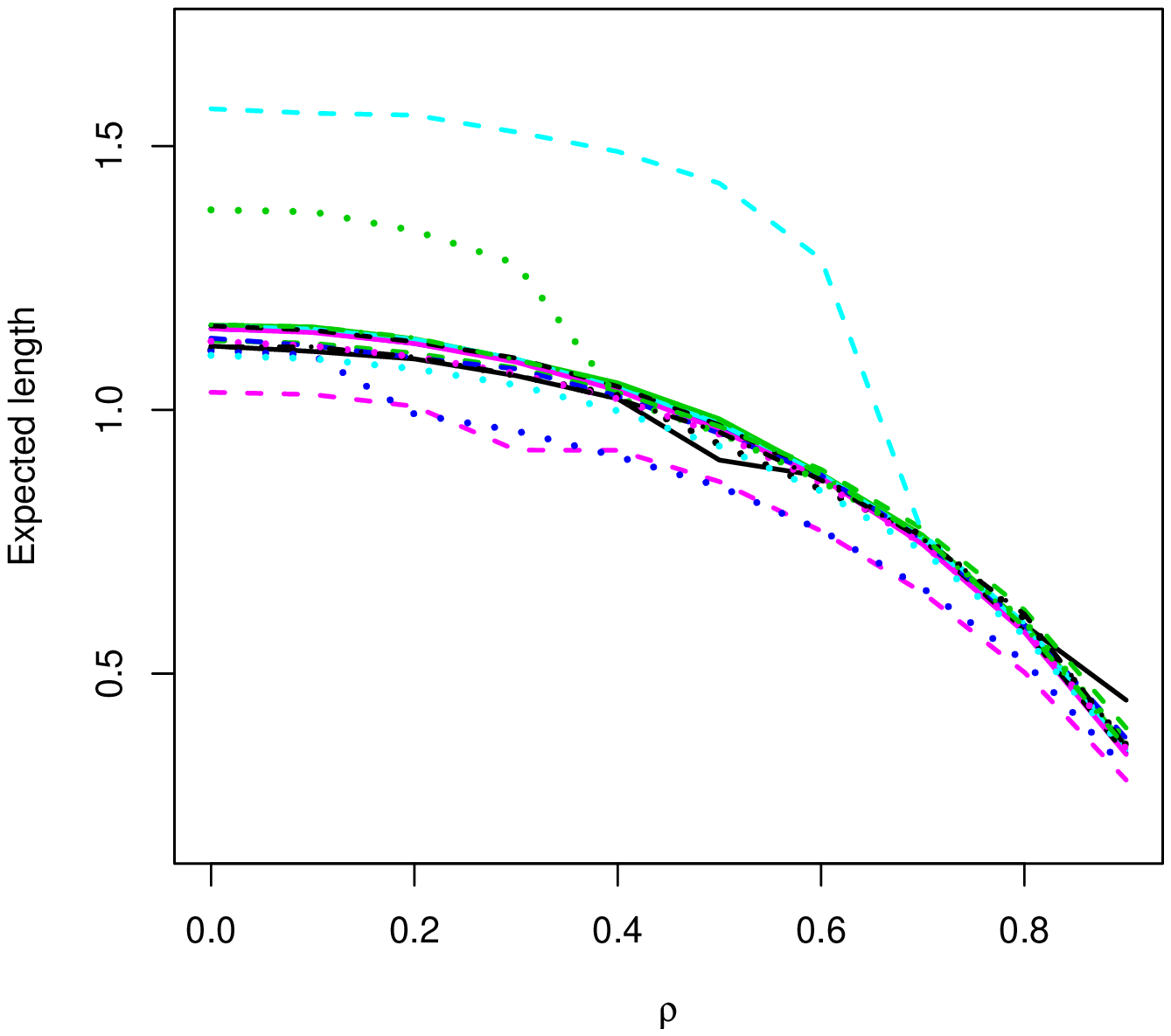}
\includegraphics[scale=0.56]{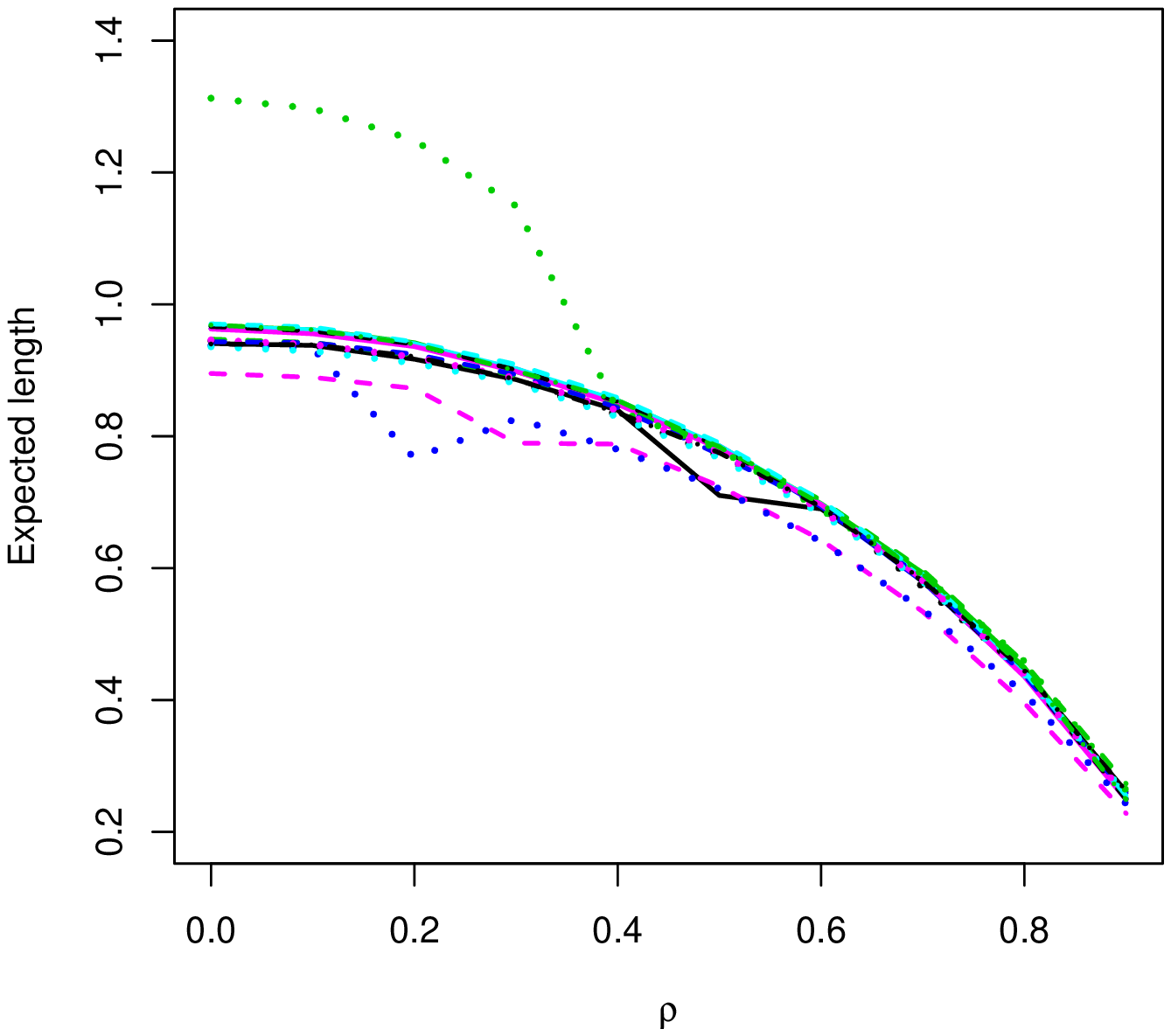}
\includegraphics[scale=0.56]{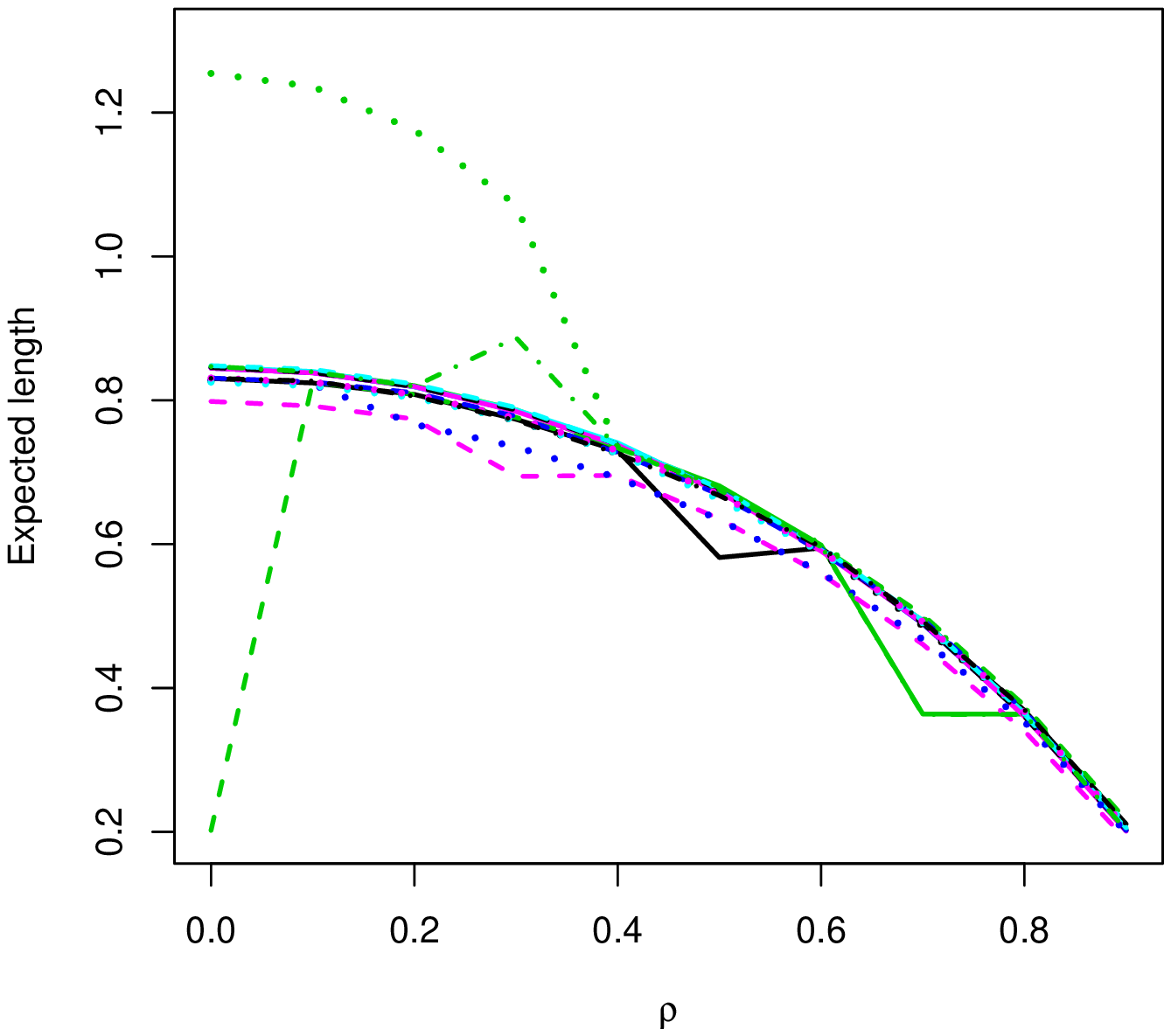}
\includegraphics[scale=0.7]{legned.eps}
\vspace{-5cm}
\caption{The expected length of confidence intervals for $\rho$ with $n=5$ (top left), $n=10$ (top right), $n=15$ (bottom left), and $n=20$ (bottom right).}\label{fig.len}
\end{figure}

\begin{table}
\begin{center}
\caption{The coverage probabilities and expected lengths of the 95\% confidence intervals (generated from a bivariate t distribution).
}\label{tab.simt}

\begin{tabular}{|l|c|c|c|c||c|c|c|c|} \hline
 & \multicolumn{8}{|c|}{Coverage Probability} \\ \hline
 & \multicolumn{4}{|c||}{$\rho =0$} & \multicolumn{4}{|c|}{$\rho =0.6$} \\ \hline
Method        \qquad   \qquad \qquad               $n$ & 3 & 5 & 10 & 25 & 3 & 5 & 10 & 25 \\ \hline
Exact  & --- & 0.9278 & 0.8998 & 0.8708 & --- & 0.9304 & 0.9106 & 0.8760 \\
z -transform & --- & 0.9285  & 0.8981  & 0.8697  & --- & 0.9355  & 0.9035  & 0.8747  \\
Hotelling $Z_1$ & 0.8977 & 0.9142 & 0.8938 & 0.8688 & 0.8991 & 0.9108 & 0.9033 & 0.8727 \\
Hotelling $Z_2$ & 0.9368 & 0.9188 & 0.8974 & 0.8709 & 0.9276 & 0.9287 & 0.8987 & 0.8718 \\
Hotelling $Z_3$ & 0.8761 & 0.9087 & 0.8935 & 0.8757 & 0.8722 & 0.9050 & 0.8987 & 0.8681 \\
Hotelling $Z_4$ & 0.9196 & 0.9177 & 0.9025 & 0.8695 & 0.9112 & 0.9165 & 0.8970 & 0.8648 \\
Ruben  & --- & 0.9710 & 0.9029 & 0.8671 & --- & 0.9651 & 0.9128 & 0.8776 \\
Muddapur 1 & 0.9380 & 0.9225 & 0.8951 & 0.8499 & 0.9263 & 0.8400 & 0.6404 & 0.3113 \\
Muddapur 2  & 0.9426 & 0.9283 & 0.9011 & 0.8716 & 0.9442 & 0.9309 & 0.9071 & 0.8704 \\
signed log-LR & 0.6396 & 0.7902 & 0.8486 & 0.8514 & 0.6436 & 0.7952 & 0.8515 & 0.8610 \\
Modified signed log-LR & 0.8741 & 0.9084 & 0.9011 & 0.8714 & 0.9212 & 0.9232 & 0.9124 & 0.8862 \\
Krishnamoorthy & 0.9450 & 0.9267 & 0.8997 & 0.8716 & 0.9422 & 0.9145 & 0.8248 & 0.5483 \\
Withers-Nadarajah 1  & 0.8018  & 0.7877  & 0.6810  & 0.4378  & 0.8302  & 0.8368  & 0.8036  & 0.6626  \\
Withers-Nadarajah 2  & 0.8903  & 0.8636  & 0.7409  & 0.4521  & 0.9023  & 0.8894  & 0.8326  & 0.6929  \\
Haddad and Provost & 0.9837  & 0.9545  & 0.9232  & 0.8791  & 0.9803  & 0.9564  & 0.9200  & 0.8778  \\
Our Generalized method & 0.7730 & 0.8710 & 0.8770 & 0.8480 & 0.7621 & 0.8740 & 0.8970 & 0.8390 \\
PB & 0.9342  & 0.9215 & 0.8972 & 0.8711 & 0.9362  & 0.9249 & 0.9029 & 0.8748 \\ \hline
 & \multicolumn{8}{|c|}{Expected Length} \\ \hline
 & \multicolumn{4}{|c||}{$\rho =0$} & \multicolumn{4}{|c|}{$\rho =0.6$} \\ \hline
Method              \qquad   \qquad \qquad            $n$ & 3 & 5 & 10 & 25 & 3 & 5 & 10 & 25 \\ \hline
Exact  & --- & 1.4362 & 1.0880 & 0.7334 & --- & 1.2428 & 0.8490 & 0.5185 \\
z-transform & --- & 1.5445 & 1.1299 & 0.7438 & --- & 1.3283 & 0.8632 & 0.5187 \\
Hotelling $Z_1$ & 1.7071 & 1.4865 & 1.1192 & 0.7426 & 1.6084 & 1.2491 & 0.8583 & 0.5179 \\
Hotelling $Z_2$ & 1.8213 & 1.5156 & 1.1241 & 0.7435 & 1.7364 & 1.2902 & 0.8607 & 0.5193 \\
Hotelling $Z_3$ & 1.6664 & 1.4719 & 1.1160 & 0.7443 & 1.5431 & 1.2343 & 0.8530 & 0.5182 \\
Hotelling $Z_4$ & 1.7900 & 1.5089 & 1.1262 & 0.7430 & 1.6942 & 1.2893 & 0.8581 & 0.5161 \\
Ruben  & --- & 1.5166 & 1.1024 & 0.7343 & --- & 1.3625 & 0.8643 & 0.5247 \\
Muddapur 1 & 1.7142 & 1.2177 & 0.8026 & 0.4865 & 1.6554 & 1.0546 & 0.6302 & 0.3644 \\
Muddapur 2  & 1.8324 & 1.5343 & 1.1330 & 0.7462 & 1.7697 & 1.3179 & 0.8711 & 0.5201 \\
signed log-LR & 1.1123 & 1.1799 & 1.0021 & 0.7100 & 0.9457 & 0.9447 & 0.7532 & 0.4929 \\
Modified signed log-LR & 1.5857 & 1.3922 & 1.0805 & 0.7324 & 1.3254 & 1.2196 & 0.8862 & 0.5194 \\
Krishnamoorthy & 1.6639 & 1.5240 & 1.3318 & 1.1518 & 1.6016 & 1.2854 & 0.9045 & 0.4731 \\
Withers-Nadarajah 1  & 1.2721 & 1.1393 & 0.8989 & 0.6179 & 1.0079 & 0.8280 & 0.5974 & 0.3621 \\
Withers-Nadarajah 2  & 1.5227 & 1.3049 & 0.9814 & 0.6420 & 1.2642 & 0.9699 & 0.6436 & 0.3783 \\
Haddad and Provost & 1.8998 & 1.6225 & 1.0306 & 0.7696 & 1.8483 & 0.9536 & 0.6005 & 0.6725 \\
Our Generalized method & 1.5054 & 1.4429 & 1.1180 & 0.7401 & 1.3285 & 1.1935 & 0.8346 & 0.5149 \\
PB & 1.8111 & 1.5240 & 1.1325 & 0.7453 & 1.7330 & 1.2989 & 0.8631 & 0.5191 \\ \hline
\end{tabular}
\end{center}
\end{table}

$ \ $

\newpage

\begin{table}
\begin{center}
\caption{ The coverage probabilities and expected lengths of the 95\% confidence intervals  (generated from a bivariate log-normal distribution).}\label{tab.simlog}
\begin{tabular}{|l|c|c|c|c||c|c|c|c|} \hline
 & \multicolumn{8}{|c|}{Coverage Probability} \\ \hline
 & \multicolumn{4}{|c||}{$\rho =0$} & \multicolumn{4}{|c|}{$\rho =0.6$} \\ \hline
Method           \qquad   \qquad \qquad                $n$ & 3 & 5 & 10 & 25 & 3 & 5 & 10 & 25 \\ \hline
Exact  & --- & 0.9478 & 0.9468 & 0.9471 & --- & 0.9506 & 0.9495 & 0.9456 \\
z -transform & --- & 0.9173 & 0.9253 & 0.9387 & --- & 0.7492  & 0.6497  & 0.5399  \\
Hotelling $Z_1$ & 0.8765 & 0.9050 & 0.9318 & 0.9420 & 0.7804 & 0.7287 & 0.6543 & 0.5369 \\
Hotelling $Z_2$ & 0.9254 & 0.9131 & 0.9299 & 0.9415 & 0.8578 & 0.7398 & 0.6566 & 0.5454 \\
Hotelling $Z_3$ & 0.8694 & 0.9097 & 0.9265 & 0.9392 & 0.7566 & 0.7124 & 0.6577 & 0.5426 \\
Hotelling $Z_4$ & 0.8947 & 0.9129 & 0.9284 & 0.9386 & 0.8227 & 0.726 & 0.6478 & 0.5483 \\
Ruben  & --- & 0.9554 & 0.9281 & 0.9386 & --- & 0.8270 & 0.6676 & 0.5495 \\
Muddapur 1 & 0.9254 & 0.9006 & 0.8861 & 0.8558 & 0.8081 & 0.5984 & 0.4076 & 0.1587 \\
Muddapur 2  & 0.9380 & 0.9236 & 0.9326 & 0.9420 & 0.8807 & 0.7421 & 0.6554 & 0.5456 \\
signed log-LR & 0.6587 & 0.8402 & 0.9150 & 0.9393 & 0.6608 & 0.8399 & 0.9117 & 0.9320 \\
Modified signed log-LR & 0.8954 & 0.9271 & 0.9325 & 0.9401 & 0.8673 & 0.9154 & 0.9472 & 0.9353 \\
Krishnamoorthy & 0.9350 & 0.9151 & 0.9256 & 0.9387 & 0.8627 & 0.7437 & 0.6447 & 0.5210 \\
Withers-Nadarajah 1  & 0.7936  & 0.8263  & 0.8400  & 0.8454  & 0.5819  & 0.5538  & 0.4673  & 0.3332  \\
Withers-Nadarajah 2  & 0.8621  & 0.8763  & 0.8720  & 0.8472  & 0.7059  & 0.6560  & 0.5286  & 0.3500  \\
Haddad and Provost & 0.9717  & 0.8996  & 0.7400  & 0.3827  & 0.9669  & 0.9041  & 0.7563  & 0.4357  \\
Our Generalized method & 0.8010 & 0.8860 & 0.9240 & 0.9320 & 0.7520 & 0.8210 & 0.6970 & 0.5380 \\
PB & 0.9253 & 0.9100 & 0.9251 & 0.9381 & 0.8355 & 0.7351 & 0.6498 & 0.5407 \\ \hline
 & \multicolumn{8}{|c|}{Expected Length} \\ \hline
 & \multicolumn{4}{|c||}{$\rho =0$} & \multicolumn{4}{|c|}{$\rho =0.6$} \\ \hline
Method               \qquad   \qquad \qquad           $n$ & 3 & 5 & 10 & 25 & 3 & 5 & 10 & 25 \\ \hline
Exact  & --- & 1.4621 & 1.1186 & 0.7494 & --- & 1.2839 & 0.8701 & 0.5265 \\
z -transform & --- & 1.5579 & 1.1865 & 0.7612 & --- & 1.2980 & 0.9223 & 0.6134 \\
Hotelling $Z_1$ & 1.6864 & 1.4973 & 1.1531 & 0.7607 & 1.4989 & 1.2395 & 0.9136 & 0.6132 \\
Hotelling $Z_2$ & 1.8064 & 1.5282 & 1.1544 & 0.7612 & 1.6577 & 1.2661 & 0.9239 & 0.6159 \\
Hotelling $Z_3$ & 1.6624 & 1.4963 & 1.1472 & 0.7603 & 1.4482 & 1.2158 & 0.9135 & 0.6127 \\
Hotelling $Z_4$ & 1.7513 & 1.5235 & 1.1545 & 0.7599 & 1.5977 & 1.2628 & 0.9286 & 0.6255 \\
Ruben  & --- & 1.5274 & 1.1235 & 0.7497 & --- & 1.3368 & 0.9146 & 0.6128 \\
Muddapur 1 & 1.0285 & 0.4395 & 0.1527 & 0.0338 & 0.9343 & 0.3503 & 0.1098 & 0.0252 \\
Muddapur 2  & 1.8256 & 1.5572 & 1.1645 & 0.7630 & 1.6914 & 1.2850 & 0.9277 & 0.6169 \\
signed log-LR & 1.1442 & 1.2251 & 1.0381 & 0.7278 & 0.9670 & 0.9768 & 0.7685 & 0.5011 \\
Modified signed log-LR & 1.6812 & 1.4225 & 1.1223 & 0.7523 & 1.5359 & 1.2280 & 0.8670 & 0.5267 \\
Krishnamoorthy & 1.6517 & 1.5037 & 1.3772 & 1.2297 & 1.5231 & 1.2456 & 1.0262 & 0.6197 \\
Withers-Nadarajah 1  & 1.2671 & 1.1719 & 0.9703 & 0.6968 & 0.9692 & 0.8795 & 0.7305 & 0.5450 \\
Withers-Nadarajah 2  & 1.4962 & 1.3333 & 1.0556 & 0.7211 & 1.2093 & 1.0378 & 0.7995 & 0.5633 \\
Haddad and Provost & 1.7281 & 1.3602 & 0.9709 & 0.6320 & 1.7146 & 1.3439 & 0.9573 & 0.6223 \\
Our Generalized method & 1.5280 & 1.4654 & 1.1530 & 0.7612 & 1.2453 & 1.1578 & 0.7939 & 0.8661 \\
PB & 1.8021 & 1.5397 & 1.1872 & 0.7628 & 1.6239 & 1.2739 & 0.9234 & 0.6142 \\ \hline
\end{tabular}
\end{center}
\end{table}

$ \ $

\newpage

\begin{table}
\begin{center}
\caption{ The 95\% confidence intervals for the examples 1 and 2.} \label{tab.ex}

\begin{tabular}{|l|c|c|c|} \hline
        &            & Example 2           &    Example 2  \\
 method &  Example 1 &  (4th Grade Score)  &   (8th Grade Score) \\ \hline
Exact  & -0.913 , -0.447 & 0.897 , 0.993 & 0.890 , 0.992 \\
z-transform & -0.919 , -0.461 & 0.905 , 0.994 & 0.899 , 0.993 \\
Hotelling $Z_1$ & -0.919 , -0.463 & 0.862 , 0.996 & 0.853 , 0.996 \\
Hotelling $Z_2$ & -0.919 , -0.463 & 0.907 , 0.994 & 0.901 , 0.993 \\
Hotelling $Z_3$ & -0.918 , -0.465 & 0.909 , 0.994 & 0.903 , 0.993 \\
Hotelling $Z_4$ & -0.918 , -0.464 & 0.909 , 0.994 & 0.903 , 0.993 \\
Ruben  & -0.915 , -0.440 & 0.888 , 0.993 & 0.881 , 0.993 \\
Muddapur 1 & -0.010 ,  -0.004 & 0.899 , 0.992 & 0.732 , 0.897 \\
Muddapur 2 & -0.920 , -0.459 & 0.905 , 0.993 & 0.899 , 0.993 \\
signed log-LR & -0.912 , -0.494 & 0.920 , 0.993 & 0.914 , 0.992 \\
Modified signed log-LR & -0.913 , -0.450 & 0.909 , 0.993 & 0.901 , 0.992 \\
Krishnamoorthy & -0.913 , -0.448 & 0.897 , 0.993 & 0.890 , 0.992 \\
Withers-Nadarajah 1  &  0.067 ,  0.781 & 0.999 , 1.000 & 0.999 , 1.000 \\
Withers-Nadarajah 2 ) &  0.037 ,  0.792 & 0.999 , 1.000 & 0.999 , 1.000 \\
Haddad and Provost &  -0.477 , 0.487 & 0.981 , 0.999 & 0.993 , 1.000 \\
Our Generalized method & -0.924 , -0.484 & 0.919 , 0.994 & 0.913 , 0.994 \\
PB & -0.919 , -0.461 & 0.906 , 0.994 & 0.900 , 0.993 \\ \hline
\end{tabular}
\end{center}
\end{table}

$ \ $
\newpage

\section{ Conclusion}
\label{sec.con}
The correlation coefficient is an important parameter in a bivariate normal distribution. In this paper, we have evaluated and compared the coverage probabilities and expected lengths of seventeen confidence intervals for this parameter via a simulation study. For sample size larger than 10, the coverage probabilities of most methods are satisfactory when the data has a bivariate normal distribution or when the data has a bivariate log-normal distribution with small variances. As opposed to other methods, for all sample size, the coverage probabilities of the exact method, Fisher's z-transformation method, Krishnamoorthy's generalized confidence interval, PB confidence interval and Withers and Nadarajah's approaches are close to the confidence coefficient. The exact method is very difficult to implement and requires solving a complex integral, and the Fisher's z-transformation method is not applicable for $n=3$. Therefore, our suggestions are Krishnamoorthy's generalized confidence interval, PB confidence interval, and Withers and Nadarajah's approaches for the correlation coefficient. However, none of the existing confidence intervals is satisfactory when the data follow a bivariate log-normal distribution. Therefore, further investigation could be done to find a confidence interval for this case.

\section*{Acknowledgment}
The author would like to thank the editor and reviewers for many constructive suggestions.

%
%


%
%
%
%
%
%
\end{document}